\documentclass[11pt]{article}

\usepackage{epsfig}
\usepackage{amsfonts}
\usepackage{amssymb}
\usepackage{amstext}
\usepackage{amsmath}
\usepackage{xspace}
\usepackage{amsthm}
\usepackage{algorithm,algorithmic}
\usepackage{fullpage}
\usepackage[compact]{titlesec}
\usepackage{algorithm}
\usepackage{enumerate}

\usepackage{xstring}


\usepackage{typearea}
\paperwidth 8.5in \paperheight 11in \typearea{15}

\newcommand{\ignore}[1]{}
\allowdisplaybreaks

\newtheorem{theorem}{Theorem}
\newtheorem{definition}[theorem]{Definition}
\newtheorem{lemma}[theorem]{Lemma}
\newtheorem{claim}[theorem]{Claim}
\newtheorem*{claim*}{Claim}

\renewenvironment{proof}{

\noindent{\bf Proof:}} {\hfill$\blacksquare$
}

\newcommand{\initOneLiners}{%
    \setlength{\itemsep}{0pt}
    \setlength{\parsep }{0pt}
    \setlength{\topsep }{0pt}
}

\newenvironment{proofof}[1]{

\bigskip\noindent{\bf Proof of {#1}:}}
{\hfill$\blacksquare$
}

\newcommand{\sse}{\subseteq}

\def\set#1{\left\{#1\right\}}

\def\gain{\mathrm{gain}}
\def\state#1{\left<#1\right>}

\DeclareMathOperator*{\E}{\mathbb{E}}
\def\Probability#1{\Pr\left[#1\right]}
\def\Expectation#1{\E\left[#1\right]}

\def\sep{\;|\;}

\def\stochTSP{\textsc{Stochastic $k$-TSP}\xspace}

\def\adTSP{\textsc{Ad-kTSP}\xspace}
\def\naTSP{\textsc{NonAd-kTSP}\xspace}

\def\cD{\ensuremath{{\cal D}}\xspace}
\def\ALG{\textsc{ALG}\xspace}
\def\OPT{\textsc{OPT}\xspace}

\def\script#1{\mathcal{#1}}
\def\sT{\script{T}}

\title{Approximation Algorithms for Stochastic $k$-TSP}
\date{\today}
\author{Alina Ene\thanks{Computer Science Department, Boston University. aene@bu.edu} \and Viswanath Nagarajan\thanks{Industrial and Operations Engineering Department, University of Michigan. viswa@umich.edu} \and Rishi Saket\thanks{IBM Research India. rissaket@in.ibm.com}}

\begin{document}
\maketitle

\begin{abstract}
We consider the stochastic $k$-TSP problem where rewards at vertices are random and the objective is to minimize the expected length of a tour that collects reward $k$. We present an adaptive $O(\log k)$-approximation algorithm, and a non-adaptive $O(\log^2k)$-approximation algorithm.  We also show that the adaptivity gap of this problem is between $e$ and $O(\log^2k)$. 
\end{abstract}

\section{Introduction}
We consider the {\em stochastic $k$-TSP} problem. There is a
metric $(V,d)$ with depot $r\in V$. Each vertex $v\in V$ has an independent 
stochastic reward $R_v \in \mathbb{Z}_+$. All reward distributions are known upfront, but the actual reward instantiation $R_v$ is only known when vertex $v$ is visited.  Given a target value $k$, the goal is to
find an adaptive tour originating from $r$ that collects a total reward at least $k$ with the
minimum expected length.\footnote{If the total instantiated  reward $\sum_{v\in V} R_v$ happens to be less than $k$,  the tour ends after visiting all vertices. }  We assume that all rewards are  supported on $\{0,1,\ldots,k\}$. 

Any feasible solution to this problem can be described by  a decision tree where nodes correspond to vertices that are visited and branches correspond  to observed reward instantiations. The size of such  decision trees can be exponentially large. So we focus on obtaining solutions (aka policies) that implicitly specify decision  trees. These solutions are called {\em adaptive} because the choice of the next vertex to visit depends on past random instantiations.    

We will also consider the special class of {\em non-adaptive} solutions. Such a solution is described simply by an ordered list of vertices: the policy involves visiting vertices in the given order until the target of $k$ is met (at which point the tour returns to $r$). These solutions are often preferred over adaptive solutions as they are easier to implement. However, the performance of non-adaptive solutions may be much worse, and so it is important to bound the {\em adaptivity gap}~\cite{DGV08} which is the worst-case ratio  between optimal adaptive and non-adaptive policies. This approach via non-adaptive policies has been very useful for a number of stochastic optimization problems, eg.~\cite{DGV08,GV06,GuhaM09,BGLMNR10,GKMR11,GKNR15,BN15,GN13}.

When the metric $(V,d)$  is a weighted-star rooted at $r$, we obtain the 
{\em stochastic knapsack cover} problem: given $n$ items with each item $i\in[n]$ having deterministic cost $d_i$ and independent random reward $R_i\in \mathbb{Z}_+$, and a target $k$, find an adaptive policy that obtains total reward $k$ at the minimum expected cost. An adaptive $2$-approximation algorithm for this problem was obtained in~\cite{DHK16}, but no bounds on the adaptivity gap were known previously.

Our main result is the following:
\begin{theorem}
There is an adaptive algorithm for stochastic $k$-TSP with approximation ratio $O(\log k)$. Moreover, the adaptivity gap is upper bounded by $O(\log^2 k)$. 
\end{theorem}
The adaptivity gap is constructive: we give a non-adaptive algorithm that is an $O(\log^2 k)$-approximation to the optimal adaptive policy. We also show that the adaptivity gap is at least $e\approx 2.718$, which holds even in the special case of stochastic knapsack cover with a single random item. 

For the deterministic $k$-TSP problem, there is a $2$-approximation algorithm~\cite{G05}. So obtaining an $O(1)$-approximation algorithm for stochastic $k$-TSP is an interesting open question. 
\section{Stochastic $k$-TSP Algorithm}
Our algorithm for \stochTSP relies on iteratively solving instances of the (deterministic) orienteering problem. In the orienteering problem, we are given a metric $(V,d)$, depot $r$, profits at vertices and a length bound $B$; the goal is to find a tour  originating from $r$ of length at most $B$ that maximizes the total profit. There is a $2+\epsilon$-approximation algorithm for orienteering~\cite{CKP12} which we can use directly. 
The vertex-profits in the orienteering instance  are chosen to be  truncated expectations of the random rewards $R_v$ where  the truncation threshold is  the current residual target. The length bounds are geometrically increasing over different iterations.   However we need to solve $\log k$ many orienteering instances  with the same length bound, which (roughly speaking) results in the $O(\log k)$ approximation ratio. This turns out to be a somewhat subtle issue because reducing the number of repetitions results in a much worse approximation ratio (see Example~2 below).

The approximation algorithm is given below. We assume (by scaling) that the minimum positive distance in the metric $(V,d)$ is one. At any point in the algorithm, $S$ denotes the set of currently visited vertices, $\sigma$ denotes the reward instantiations of $S$, and $k(\sigma)$ is the total observed reward.  
The number of iterations of the inner for-loop is $\alpha:=c\cdot H_k$, where $c\ge 1$ is a constant to be fixed later and $H_k\approx \ln k$ is the $k$-th	harmonic number.  We will show that \adTSP is an $8\alpha$-approximation algorithm for \stochTSP. Some of the ideas in this analysis are similar to those used previously in deterministic routing problems~\cite{CGRT03,FHR07} and  stochastic covering problems~\cite{INZ12}.

\begin{algorithm}[h!] \caption{Algorithm \adTSP\label{alg:kTSP}} 
  \begin{algorithmic}[1]
    \STATE initialize $\sigma\gets \emptyset$,	$S\gets \emptyset$ and $k(\sigma)=0$.
	\FOR{phase $i=0,1,\ldots$}  
	\FOR{$t=1,\ldots,\alpha$}
		\STATE define profits at vertices as follows
		$$w_v := \left\{ 
		\begin{array}{ll}
		\E[\min\{ R_v,k-k(\sigma)\}] & \forall  \, v\in V\setminus S \\
		0 & \forall \, v\in S
\end{array}\right.$$
		\STATE \label{step:orient} using a $\rho$-approximation algorithm for the {\em orienteering} problem, compute a tour $\pi$ originating from $r$ of length at most $2^i$  with maximum total profit. 
		\STATE traverse tour $\pi$ and observe the actual rewards on
		it; augment $S$ and $\sigma$ accordingly.
		\STATE if $k(\sigma)\ge k$ or all vertices have been visited, the solution ends. 
    \ENDFOR
    \ENDFOR
  \end{algorithmic}
\end{algorithm}



We view each iteration of the outer for loop as a \emph{phase} and  use $i$ to index the phases. For any phase $i \geq 0$, we define the following quantities:
\begin{eqnarray*}
	u_i &:=& \Probability{\mbox{\adTSP continues beyond phase }	i}\\
	u^*_i &:=& \Probability{\mbox{optimal policy continues beyond distance } 2^i}
\end{eqnarray*}

Since the minimum distance in the metric is one we have $u^*_0=1$. 
The following lemma is the main component of the analysis.

\begin{lemma}\label{thm:stoch-tsp-main}
	For any phase $i \ge 1$, we have $u_i \le {u_{i-1} \over 4} +
	u^*_i$. 
\end{lemma}

Using Lemma~\ref{thm:stoch-tsp-main}, we can finish the
analysis as follows.
	Let \ALG denote the expected length of the tour constructed by
	the \adTSP algorithm. Let \OPT denote the expected length of
	the optimal adaptive tour. The total distance traveled by the \adTSP
	algorithm in the first $i$ phases is at most
		 $\alpha \sum_{j = 0}^i 2^j \le 2^{i + 1} \alpha$. 
Using this  we obtain $\ALG \leq 2\alpha \sum_{i \geq 1} 2^i u_i + 4\alpha$. 	Also,
$\OPT \geq \frac12 \sum_{i \geq 1} 2^i u^*_i + u^*_0 = \frac12 \sum_{i \geq 1} 2^i u^*_i + 1$.	Letting $T:=
\sum_{i \geq 1} 2^i u_i$ and using Lemma~\ref{thm:stoch-tsp-main}, we obtain 
$$T \,\, \le  \,\, \frac14 \sum_{i\ge 1} 2^i\cdot u_{i-1} +  \sum_{i\ge 1} 2^i\cdot u_i^*  \,\, \le  \,\, \frac12 \sum_{i\ge 1} 2^i\cdot u_{i} +  \sum_{i\ge 1} 2^i\cdot u_i^* +\frac12 \,\, \le  \,\, \frac12 T + 2\cdot \OPT-\frac32.$$
It follows that $T\le 4\cdot \OPT-3$ and $\ALG\le 8\alpha\cdot \OPT$. Thus we obtain an $8\alpha$-approximation algorithm.

\medskip\noindent
Now we turn to the proof of Lemma~\ref{thm:stoch-tsp-main}. We start by introducing some
notation. Consider an arbitrary point in the execution of algorithm \adTSP, and let
$\left<S, \sigma, k(\sigma)\right>$ be a triple in which $S$ is the
set of vertices visited so far, $\sigma$ is the observed 
instantiation of $S$, and $k(\sigma)$ is the total
reward observed in $S$. We refer to such a triple $\left<S, \sigma,
k(\sigma)\right>$ as the \emph{state} of the algorithm.


Here is an outline of the proof. First, in  Lemma~\ref{lem:opt-to-orient-tour} we relate the optimal value of the   
orienteering  instance in any iteration to the optimal adaptive solution. Then in Lemma~\ref{lem:additional-reward} we show that the expected increase in reward (relative to the residual target)  in any iteration  with length bound $2^i$  is at least a constant fraction of the (conditional) probability that the optimal adaptive solution completes by time $2^i$. This is used in Claim~\ref{cl:gain-lb} to lower bound  the total  increase in reward (relative to the residual target) in phase $i$. Combined with an upper bound on this increase (Claim~\ref{cl:gain-ub}) we complete the proof of Lemma~\ref{thm:stoch-tsp-main}. 

\begin{lemma}\label{lem:opt-to-orient-tour}
	Consider a state $\left<S, \sigma, k(\sigma) \right>$ of the
	\adTSP algorithm. Consider the \textsc{Orienteering} instance
	in which each vertex in $S$ is assigned a profit of zero and each
	vertex $v$ not in $S$ is assigned a profit of
	$\Expectation{\min\set{R_v, k - k(\sigma)}}$. There is an
	orienteering tour from $r$ of length at most $2^i$ with profit at
	least $(k-k(\sigma))\cdot p^*_i(\sigma)$, where
		$$p^*_i(\sigma) = \Probability{\mbox{ optimal policy
		completes before distance }2^i \sep \sigma\,}.$$
\end{lemma}
\begin{proof} 
Consider the tree $\sT^*$ representing the optimal policy. We {\em condition} on the instantiations $\sigma$ on vertices $S$ to obtain tree $\sT^*(S, \sigma)$.   Formally, we start with $\sT^*(S, \sigma) = \sT^*$ and apply the following transformation for each vertex $v\in S$ with instantiation $\sigma_v$: at each node
$\nu$ in $\sT^*(S, \sigma)$ corresponding to $v$, we remove all subtrees of $\nu$ except the subtree corresponding to instantiation $\sigma_v$; additionally, the probability of this edge (labeled $\sigma_v$) is set to one. Note that the probabilities at nodes corresponding to vertices $V\setminus S$ are unchanged, since rewards at different vertices are independent.

Now, {\em mark} those nodes in $\sT^*(S, \sigma)$ that correspond to completion (i.e. reward at least $k$ is collected) within distance $2^i$. Note that the probability of reaching a marked node is exactly $p^*_i(\sigma)$. Finally, let $\sT$ denote the subtree of $\sT^*(S, \sigma)$ containing only those nodes that have a marked descendant. When tree $\sT$ is traversed, the probability of reaching a leaf-node is $p^*_i(\sigma)$; with the remaining probability the traversal ends at some internal node.  Clearly, every leaf-node in $\sT$ is marked and hence corresponds to an $r$-tour of length at most $2^i$ along with instantiated rewards that sum to at least $k$. Define modified rewards at nodes of $\sT$ as follows:
$$\hat{R}_v = \left\{
\begin{array}{ll}
\min\{R_v,k-k(\sigma)\}& \mbox{ if }v\not\in S\\
0&\mbox{ if }v \in S
\end{array}\right.$$
Notice that the profits in the deterministic \textsc{Orienteering} instance are precisely $w_v=\E[\hat{R}_v]$. Observe that the sum of modified rewards at each leaf of $\sT$ is at least $k-k(\sigma)$. Thus the expected $\hat{R}$-reward obtained in $\sT$ is $\hat{E}\ge (k-k(\sigma))\cdot p^*_i(\sigma)$. Also,
$$\hat{E} = \sum_{\nu\in \sT} \Pr[\mbox{reach }\nu]\cdot \E\left[ \hat{R}_\nu \,|\, \mbox{reach }\nu\right] = \sum_{\nu\in \sT} \Pr[\mbox{reach }\nu]\cdot \E\left[ \hat{R}_\nu \right] = \sum_{\nu\in \sT} \Pr[\mbox{end at }\nu]\cdot \left(\sum_{\mu\preceq \nu} w_\mu\right).$$
Above, for a node $\nu\in \sT$, we use $\hat{R}_\nu$ and $w_\nu$ to denote the respective variables for the vertex (in $V$) corresponding to $\nu$. Also, the notation $\mu\preceq \nu$ refers to node $\mu$ being an ancestor of node $\nu$ in $\sT$. The first equality is by definition of $\hat{E}$, the second uses the fact that $\{\hat{R}_v:v\in V\}$ are independent, and the last is an interchange of summation. By averaging, there is some node $\nu'\in\sT$ with $\sum_{\mu\prec \nu'} w_\mu\ge \hat{E}\ge (k-k(\sigma))\cdot p^*_i(\sigma)$. We can assume  that $\nu'$ is a leaf node (otherwise we can reset $\nu'$ to be any leaf node of $\sT$ below $\nu'$). So the $r$-tour corresponding to this node $\nu'$ has length at most $2^i$ and is feasible for the  deterministic \textsc{Orienteering} instance. Moreover, it has profit at least  $(k-k(\sigma))\cdot p^*_i(\sigma)$ as claimed. 
\end{proof}

\begin{lemma} \label{lem:additional-reward}
	Consider a state $\left<S, \sigma, k(\sigma) \right>$ of the
	\adTSP algorithm with $k(\sigma)<k$. Consider the \textsc{Orienteering} instance
	with profits $\{w_v:v\in V\}$ where $w_v=0$ for $v\in S$ and $w_v=\Expectation{\min\set{R_v, k - k(\sigma)}}$ for $v\not\in S$. Let $\pi$ be any $\rho$-approximate
	orienteering tour consisting of vertices $V(\pi)$. Then,
		$$\E\left[ \min\set{\sum_{v \in V(\pi) \setminus S} R_v,\, k -
		k(\sigma)} \right] \quad \geq \quad {1 \over \rho} \cdotp \left(1-\frac{1}{e}\right)\cdot (k - k(\sigma))	\cdotp p_i^*(\sigma).$$ 
\end{lemma}
\begin{proof} For each $v\in V(\pi)\setminus S$ define $X_v:=\min\left\{\frac{R_v}{k-k(\sigma)},\, 1\right\}$. 
Note that $X_v$s are independent $[0,1]$ random variables, and $\E[X_v]=\frac{w_v}{k-k(\sigma)}$. Let $X:=\sum_{v\in V(\pi)\setminus S} X_v$ and $Y=\min(X,1)$.   So the profit obtained by solution $\pi$ to the deterministic orienteering instance is $(k-k(\sigma))\cdot\E[X]$. Using Lemma~\ref{lem:opt-to-orient-tour} and the fact that $\pi$ is a $\rho$-approximate solution, we obtain $\E[X]\ge {1 \over \rho} \cdotp p_i^*$. We can now complete the proof of the lemma by applying Theorem~\ref{th:indep-cap-exp} below.
\end{proof}

\begin{theorem}[\cite{AMMPS11}]\label{th:indep-cap-exp}
Given a set $\{X_v\}$ of independent $[0,1]$ random variables with $X=\sum X_v$ and $Y=\min(X,1)$, we have $\E[Y]\ge (1-1/e)\cdot \min\{\E[X],1\}$.
\end{theorem}

Now we are ready to prove Lemma~\ref{thm:stoch-tsp-main}.
\begin{proofof}{Lemma~\ref{thm:stoch-tsp-main}} Consider any phase $i\ge 1$ and let $\state{S, \sigma,  k(\sigma)}$ be the state of the \adTSP algorithm at the start of some iteration of the inner loop. Let $\pi$ be the orienteering tour that the algorithm visits next. Define
		$$\gain(\state{S, \sigma}) \, := \,{\E\left[\min\set{\sum_{v \in V(\pi) \setminus S} R_v, \, k - k(\sigma)}\right] \over k -k(\sigma)}  \,\, \mbox{ if } k(\sigma)<k,  $$
and $\gain(\state{S, \sigma}) := 0$ if $k(\sigma)\ge k$.		
The quantity $\gain(\state{S, \sigma})$ measures the expected fraction of the residual target that we cover after visiting $\pi$.

Let $I(\sigma) = [k(\sigma) \geq k]$ be an indicator variable
	that is equal to one if $k(\sigma) \geq k$ and it is equal to
	zero otherwise. We have:
	\begin{equation} \label{eq:eq1}
		\gain(\state{S, \sigma}) \quad \geq \quad {1 \over \rho} \cdotp \left(1-\frac{1}{e}\right)\cdotp \left( p^*_i(\sigma) - I(\sigma)\right).
	\end{equation}
To see this, note that if $I(\sigma)$ is one, the gain is zero and the inequality above
	holds since $p^*_i(\sigma) \leq 1$; and if $I(\sigma)$ is zero, this inequality 
	follows from Lemma~\ref{lem:additional-reward}.

Recall that	there are $\alpha$ iterations of the inner loop in  phase $i$, and we use $t \in [\alpha]$ to index the iterations. For each iteration $t$ (of phase $i$), let
	$S_t$ be the random variable denoting the vertices visited until iteration $t$, and let $\sigma_t$ denote the reward instantiations of $S_t$. Define:
		$$G_t \quad := \quad \E_{\state{S_t, \sigma_t}} [\gain(\state{S_t,
		\sigma_t})]$$
	Let $\Delta = \sum_{t = 1}^{\alpha} G_t$ denote the total gain in phase $i$. The proof relies on  upper and lower
	bounding $\Delta$, which is done in the next two claims.

	\begin{claim}\label{cl:gain-ub}
		We have $\Delta \leq H_{k} \cdotp u_{i - 1}$.
	\end{claim}
	\begin{proof}
Let $\Delta(q)$ denote the value of $\Delta$ conditioned on the instantiations $q$ of all random variables. If \adTSP (conditioned on $q$) finishes before phase $i$ then gain in each iteration is zero and $\Delta(q)=0$.

In the following, we assume that \adTSP (conditioned on $q$) reaches phase $i$. Let $L\le k$ denote the residual target at the start of phase $i$, and $J_1,\ldots,J_\alpha\in \mathbb{Z}_+$ the incremental rewards obtained in each of the $\alpha$ iteration of phase $i$; recall that all rewards are integral. Then,
$$\Delta(q) \,\,\le \,\,\sum_{t = 1}^{\alpha} {J_t \over L -
J_1 - \cdots - J_{t -1}} \,\, \leq \,\, \sum_{t = 1}^L {1 \over t} \,\,=\,\, H_{L}\,\,\le\,\, H_{k}$$
The claim now follows since the algorithm reaches phase $i$ only with probability $u_{i - 1}$.\end{proof}

	\begin{claim}\label{cl:gain-lb} We have $\Delta \geq {\alpha \over \rho} \cdotp \left(1-\frac{1}{e}\right) \cdotp (u_i - u^*_i)$.	
	\end{claim}
	\begin{proof}
		For any $t \in [\alpha]$, by Inequality~\eqref{eq:eq1}, we have
$$G_t  \quad = \quad \E_{\state{S_t, \sigma_t}}[\gain(\state{S_t,\sigma_t})]
			\quad \geq \quad {1 \over \rho} \cdotp \left(1-\frac{1}{e}\right) \cdotp \E_{\sigma_t}[p^*_i(\sigma_t) -	I(\sigma_t)]$$
		Let $p(t)$ be the probability that \adTSP finishes by
		iteration $t$ of phase $i$. Since the probabilities $p(t)$
		are non-decreasing in $t$, we have $p(t) \leq p(\alpha) = \Pr[\adTSP \mbox{ finishes by phase }i] = 1 - u_i$. For a fixed iteration $t$, the possible outcomes of $\langle S_t, \sigma_t\rangle$ correspond to a partition of the overall sample space.  
So we have $\E_{\sigma_t}[I(\sigma_t)] = p(t) \leq 1 - u_i$ and $\E_{\sigma_t}[p^*_i(\sigma_t)] = \Pr[\mbox{optimal policy completes within distance }2^i] = 1 - u^*_i$. This completes the proof since $\Delta=\sum_{t=1}^\alpha G_t$.\end{proof}
	
It follows from the two claims that
$${\alpha \over \rho} \cdotp (1-1/e)\cdotp \left(u_{i} - u^*_i\right) \quad \leq  \quad  \Delta \quad \le \quad H_k\cdotp u_{i - 1}.$$
Setting $\alpha=4\rho\frac{e}{e-1}\cdot H_k$ implies the lemma.
\end{proofof}

\smallskip
{\sc Example 1.} The analysis of \adTSP is tight up to constant factors, even in the deterministic setting. Consider an instance of deterministic knapsack cover with $k=2^\ell$ (for large integer $\ell$) and $\ell(\ell+1)$ items as follows. For each $i\in\{0,1,\ldots,\ell\}$  there is one item of cost $2^i$ and reward $2^i$, and  $\ell-1$ items of cost $2^i$ and reward $1$. Clearly the optimal cost is $2^\ell$. However the algorithm will select in each iteration $i=0,1,\ldots,\ell-3$, all $\ell$ items of cost $2^i$: note that the reward from these items is at most $\ell^2 + \sum_{i=0}^{\ell-3} 2^i \le \ell^2+2^{\ell-2} < 2^\ell$. So the algorithm's cost is at least $\ell\cdot 2^{\ell-3}$.

\smallskip
{\sc  Example 2.} A natural variant of \adTSP is to perform the inner iterations (for each phase $i=0,1,\ldots$) a constant number of times instead of $\Theta(\log k)$. Indeed, this variant achieves a constant approximation ratio for deterministic $k$-TSP. However, as shown next, such variants perform poorly in the stochastic setting.

Consider the variant of \adTSP that performs $1\le h = o\left(\frac{\log k}{\log\log k}\right)$ iterations in each phase $i$. Choose $t\in \mathbb{Z}_+$ and set $\delta=1/(ht)$ and $k=(ht)^{2ht}\in \mathbb{Z}_+$. Note that $ht=\Theta(\log k/ \log\log k)$.
Define a star metric centered at the depot $r$, with $ht+1$ leaves $\{u_{ij} : 0\le i\le t-1, 0\le j\le h-1\}\bigcup \{w\}$. The distances are $d(r,u_{ij})=2^i$ for all $i\in[t]$ and $j\in [h]$; and $d(r,w)=1$. Each $u_{ij}$ contains three co-located items: one of deterministic reward $(1-\delta)\delta^{hi+j}\cdot k$, and two having reward $\delta^{hi+j}\cdot k$ with probability $\delta$ (and zero otherwise). Finally $w$ contains a deterministic item of reward $k$. By the choice of parameters, all rewards are integer valued. The optimal cost is $2$, just visiting vertex $w$. 

Now consider the execution of the modified \adTSP algorithm. The probability that all the random items in the $u_{ij}$-vertices have zero reward is $(1-\delta)^{2ht}\ge \Omega(1)$. Conditioned on this event, it can be seen inductively that in  the $j^{th}$ iteration of phase $i$ (for all $i$ and $j$),
\begin{itemize}
\item the total observed reward until this point is  $k(1-\delta) \sum_{\ell=0}^{hi+j-1} \delta^\ell  =  k\left(1-\delta^{hi+j}\right)$.
\item the algorithm's tour (and optimal solution to the orienteering instance) involves visiting just vertex $u_{ij}$ and choosing the three items in $u_{ij}$, for a total profit of $(1+\delta)\cdot \delta^{hi+j}\cdot k$. 
\end{itemize}

Thus the expected cost of this algorithm is $\Omega(h\cdot 2^t)$, implying an approximation ratio  $\exp\left(\frac{\log k}{h\cdot \log\log k}\right)$.

\subsection{Non-Adaptive Algorithm}
We now show that the above adaptive algorithm can be simulated in a non-adaptive manner, resulting in an $O(\log^2k)$-approximate non-adaptive algorithm. This also upper bounds the adaptivity gap by the same quantity. Algorithm \naTSP constructs the non-adaptive tour $\tau$ iteratively; $S$ denotes the set of vertices visited in the current tour. 

\begin{algorithm}[h!] \caption{Algorithm \naTSP \label{alg:kTSP-na}} 
  \begin{algorithmic}[1]
    \STATE initialize $\tau, S\gets \emptyset$.
	\FOR{phase $i=0,1,\ldots$}
	\FOR{$t=1,\ldots,\alpha$}
	\FOR{$j=0,1,\ldots,\lfloor \log_2k\rfloor$}
		\STATE define profits as follows
		$$w_v := \left\{ 
\begin{array}{ll}
		\E[\min\{ R_v, k/2^j\}] & \mbox{ for } v\in V\setminus S\\
		0& \mbox{ for } v\in S
\end{array}\right.$$ 
		\STATE \label{step:orient2} using a $\rho$-approximation algorithm for the {\em orienteering} problem, compute a tour $\pi$ originating from $r$ of length at most $2^i$  with maximum total profit. 
\STATE append tour $\pi$ to $\tau$, i.e. $\tau\gets \tau\circ \pi$.
    \ENDFOR
    \ENDFOR
    \ENDFOR
  \end{algorithmic}
\end{algorithm}

The difference from \adTSP is in the inner for loop (indexed by $j$); unlike \adTSP the non-adaptive algorithm does not know the reward accrued so far. Instead we append $\log_2k$ different orienteering tours, corresponding to the possible amounts of reward accrued. We set $\alpha:=c'\cdot H_k$, where $c'$ is a constant to be fixed later. Let $\ell:=1+\lfloor \log_2 k\rfloor$ the number of iterations of the innermost loop. 

The analysis for \naTSP is almost identical to \adTSP. We will show that this is an $(8\alpha\ell)$-approximation algorithm for \stochTSP. As before, each iteration of the outer for loop is called a \emph{phase}, which is indexed by $i$. For any phase $i \geq 0$, define
\begin{eqnarray*}
	u_i &:=& \Probability{\mbox{\naTSP continues beyond phase }	i}\\
	u^*_i &:=& \Probability{\mbox{optimal adaptive policy continues beyond distance } 2^i}
\end{eqnarray*}
Just as in Lemma~\ref{thm:stoch-tsp-main} we will show:
\begin{equation}\label{eq:na-main}
u_i \quad \le \quad {u_{i-1} \over 4} + u^*_i,\qquad \forall i\ge 1.
\end{equation} 
Since the total distance traveled by \naTSP in the first $i$ phases is at most $\ell\alpha\sum_{h=0}^i 2^h\le \ell\alpha\,2^{i+1}$, it follows (as for \adTSP) that the expected length \ALG of \naTSP is at most $8\ell\alpha\cdot \OPT$.

We now prove~\eqref{eq:na-main}. Consider a fixed phase $i\ge 1$ and one of the $\alpha$ iterations of the second for-loop (indexed by $t$). Let $S$ denote the set of vertices visited before the start of this iteration $\langle i,t\rangle$. Let $\sigma$ denote  the reward instantiations at $S$, and $k(\sigma)$ the total reward in $\sigma$. Note that $\sigma$ is only used in the analysis and not in the algorithm. Let $V(i,t)\sse V\setminus S$ be the new vertices visited in iteration $\langle i,t\rangle$; these come from  $\ell$ different subtours corresponding to the inner for-loop. We will show (analogous to Lemma~\ref{lem:additional-reward}) that:
\begin{equation}\label{eq:na-gain}
\E\left[ \min\set{\sum_{v \in V(i,t)} R_v,\, k -
		k(\sigma)} \right] \quad \geq \quad {1 \over 2\rho} \cdotp \left(1-\frac{1}{e}\right)\cdot (k - k(\sigma))	\cdotp p_i^*(\sigma).
\end{equation}
Above, $p^*_i(\sigma) = \Probability{\mbox{ optimal adaptive policy completes before distance }2^i \sep \sigma\,}$. Exactly as in the proof of Lemma~\ref{thm:stoch-tsp-main}, this would imply~\eqref{eq:na-main} when we set $\alpha=8\rho\frac{e}{e-1}\cdot H_k$.

It remains to prove~\eqref{eq:na-gain}. Let $g\in \{0,1,\ldots,\ell\}$ denote the index such that $\frac{k}{2^g}\le k-k(\sigma)<\frac{k}{2^{g+1}}$ and let $c:=k/2^g$. An identical proof to Lemma~\ref{lem:opt-to-orient-tour} implies:
\begin{lemma} \label{lem:na-orient} Consider the \textsc{Orienteering} instance with profits $s_v:=\E[\min\set{R_v, c}]$ for $v\in V\setminus S$ and $s_v :=0$ otherwise. There is an
orienteering tour of length at most $2^i$ with profit at least $c\cdot p^*_i(\sigma)$.
\end{lemma}

Let $T\sse V\setminus S$ denote the vertices added in iteration $\langle i,t\rangle$ corresponding to inner loop iterations $j\in \{0,1,\ldots,g-1\}$. Note that in the inner loop iteration $j=g$, we have profits $w_v=\E[\min\set{R_v, c}]$ for $v\in V\setminus (S\cup T)$ and $w_v=0$ otherwise. Lemma~\ref{lem:na-orient} now implies that the optimal profit of the orienteering instance in iteration $j=g$ is at least $c\cdot p^*_i(\sigma) - \sum_{u\in T} \E[\min\set{R_u, c}]$. Let $T'$ denote the vertices added in the inner-loop iteration $j=g$. Since we obtain a $\rho$-approximate solution, it follows that the additional profit obtained $\sum_{v\in T'} w_v \ge \frac1\rho \left( c\cdot p^*_i(\sigma) - \sum_{u\in T} \E[\min\set{R_u, c}]\right)$. So,
$$\sum_{v\in V(i,t)} \E[\min\set{R_v, c}] \quad \ge  \quad \sum_{v\in T\cup T'} \E[\min\set{R_v, c}] 
 \quad \ge \quad \frac{c}{\rho}\cdot p^*_i(\sigma).$$
Exactly as in Lemma~\ref{lem:additional-reward} we now obtain:
$$\E\left[ \min\set{\sum_{v \in V(\pi) \setminus S} R_v,\, c} \right] \quad \geq \quad {1 \over \rho} \cdotp \left(1-\frac{1}{e}\right)\cdot c \cdotp p_i^*(\sigma).$$ 
And using $c\le k-k(\sigma)<2c$, we obtain~\eqref{eq:na-gain}.

\smallskip
{\sc Example 3.} The analysis of \naTSP is tight up to constant factors, even in the deterministic setting. This example is similar to that for \adTSP. Consider an instance of deterministic knapsack cover with $k=2^\ell$ and $\ell^2(\ell+1)$ items as follows. For each $i\in\{0,1,\ldots,\ell\}$  there is one item of cost $2^i$ and reward $2^i$, and  $\ell^2 -1$ items of cost $2^i$ and reward $1$. Clearly the optimal cost is $2^\ell$. However algorithm \naTSP will select in each iteration $i=0,1,\ldots,\ell-3$, all $\ell^2$ items of cost $2^i$: note that the reward from these items is at most $\ell^3 + \sum_{i=0}^{\ell-3} 2^i \le \ell^3+2^{\ell-2} < 2^\ell$. So the algorithm's cost is at least $\ell^2\cdot 2^{\ell-3}$, implying an approximation ratio of $\Omega(\log^2k)$.

\smallskip
{\sc Example 4.} One might also consider the variant of \naTSP where the second for-loop (indexed by $t$) is performed $1\le h\ll \log k$ times instead of $\Theta(\log k)$. The following example shows that such variants have a large approximation ratio. Consider an instance of stochastic knapsack cover with $k=2^\ell$.  There are an infinite number of items, each of cost one and reward $k$ with probability $\frac{2}{3\ell}$ (the reward is otherwise zero). For each $i\in \{0,1,\ldots,\ell/h\}$ and $j\in\{1,\ldots,\ell\}$ there are $h$ items of cost $2^i$ and reward $k/2^j$ with probability $\frac{2}{3\ell}$ (the reward is otherwise zero). The optimal (adaptive and non-adaptive) policy considers only the cost one items with $\{0,k\}$ reward, and has optimal cost $O(\ell)$.
	
Algorithm \naTSP chooses in each iteration $i\in \{0,1,\ldots,\ell/h\}$ (of the first for-loop), iteration $t\in \{1,\cdots,h\}$ (of the second for-loop) and iteration $j\in\{0,1,\ldots,\ell\}$ (of the third for-loop) one of the items of cost $2^i$ with $\{0,k/2^j\}$ reward. (This is an optimal choice for the orienteering instance as rewards are capped at $k/2^j$ and the length bound is $2^i$.) These items have total cost $\Theta(\ell\cdot 2^{\ell/h})$. For each $j\in\{0,1,\ldots,\ell\}$ there are $\frac{\ell}h\cdot h = \ell$ items having reward $k/2^j$ with probability $\frac{2}{3\ell}$. It can be seen that the total reward from these items is less than $k$ with constant probability. So the expected cost of \naTSP is $\Omega(\ell\cdot 2^{\ell/h})$, implying an approximation ratio $\Omega(2^{\ell/h})$.
		
\section{Lower Bound on Adaptivity Gap } We show that the adaptivity gap of \stochTSP is at least $e\approx 2.718$. This holds even in the special case of the {\em stochastic covering knapsack} problem with a single random item.

The gap example is based on the following problem studied in Chrobak et al.~\cite{CKNY08}.
\begin{definition}[Online Bidding Problem] Given input $n\in \mathbb{Z}_+$, an algorithm outputs a randomized sequence $b_1,b_2,\ldots, b_\ell$ of bids from the set $[n]:=\{1,2,\ldots,n\}$. The algorithm's cost under (an unknown) threshold $T\in[n]$ equals the (expected) sum of its bid values until it bids a value at least $T$. The algorithm is $\beta$-competitive if its cost under threshold $T$ is at most $\beta\cdot T$, for all $T\in[n]$.
\end{definition}
\begin{theorem}[\cite{CKNY08}] \label{th:online-bidding}
There is no randomized algorithm for online bidding that is less than $e$-competitive.
\end{theorem}

Without loss of generality, any bid sequence must be increasing; let $\Gamma$ denote the set of increasing sequences on $[n]$. For any $I\in \Gamma$ and $T\in [n]$, let $C(I,T)$ denote the cost of sequence $I$ under threshold $T$. In terms of this notation, Theorem~\ref{th:online-bidding} is equivalent to:
\begin{equation}\label{eq:bidding-lb}
\min_{\pi \,:\, \mbox{distribution}(\Gamma)} \quad \max_{T\in [n]}\quad \frac{\E_{I\gets \pi}[C(I,T)]}{T} \quad \ge \quad e.
\end{equation}

We now define the stochastic covering knapsack instance. The target $k:=2^{n+1}$. There is one random item $r$ of zero cost having reward $k-2^i$ with probability $p_i$ (for all $i\in[n]$). There are $n$ deterministic items $\{u_i\}_{i=1}^n$ where $u_i$ has cost $i$
and reward $2^i$. We will show that there exist probabilities $p_i$s so that the adaptivity gap is $e$. 

It is clear that an optimal policy (adaptive or non-adaptive) will first choose item $r$, since it has zero cost. Moreover, an optimal adaptive policy will next choose item $u_i$ exactly when $r$ is observed to have reward $k-2^i$. Hence the optimal adaptive cost is $\sum_{i=1}^n i\cdot p_i$.

Any non-adaptive policy is given by a sequence $\tau$ of the deterministic items; recall that item $r$ is always chosen first. Moreover, due to the exponentially increasing rewards, we can assume  that $\tau$ is an increasing sub-sequence of $\{u_i: 1\le i\le n\}$. So there is a one-to-one correspondence between non-adaptive policies and $\Gamma$ (in the online bidding instance). Note that the cost of non-adaptive solution $I\in \Gamma$ is exactly $\sum_{T=1}^n p_T\cdot C(I,T)$; if the random item $r$ has size $k-2^T$ then the policy will keep choosing items in $I$ until it reaches an index at least $T$. Thus, the maximum adaptivity gap achieved by such an instance is:
\begin{equation}\label{eq:skc-ad-gap}
\max_{p\,:\,\mbox{distribution}([n])}\quad \min_{I\in \Gamma} \quad \frac{\E_{T\gets p}[C(I,T)]}{\E_{T\gets p}[T]}.
\end{equation}
The next lemma relates the quantities in~\eqref{eq:bidding-lb} and~\eqref{eq:skc-ad-gap}. Below, for any set $S$, $\cD(S)$ denotes the collection of probability distributions on $S$.
\begin{lemma} For any non-negative matrix $C(\Gamma,[n])$, we have:
$$\max_{p\in\cD([n])}\,\, \min_{I\in \Gamma} \,\, \frac{\E_{T\gets p}[C(I,T)]}{\E_{T\gets p}[T]} 
\quad =\quad \min_{\pi \in\cD(\Gamma)} \,\, \max_{T\in [n]}\,\, \frac{\E_{I\gets \pi}[C(I,T)]}{T}.$$
\end{lemma}
\begin{proof}
The second expression equals the LP:
$$\begin{array}{lll}
\min & \beta&\\
\mbox{s.t.}& \,\, T\cdot \beta  - \sum_{I\in \Gamma} C(I,T) \cdot \pi(I) \ge 0,& \quad \forall T\in[n],\\
&\sum_{I\in \Gamma} \pi(I) \ge 1,&\\
& \beta, \mathbf{\pi} \ge 0.&
\end{array}
$$
Taking the dual, we obtain:
$$\begin{array}{lll}
\max & \alpha&\\
\mbox{s.t.}& \,\, \alpha - \sum_{T\in[n]} C(I,T) \cdot \sigma(T) \le 0,& \quad \forall I\in\Gamma,\\
&\sum_{T\in [n]} T\cdot \sigma(I) \le 1,&\\
& \alpha, \mathbf{\sigma} \ge 0.&
\end{array}
$$
Define functions $g,f:[n]\rightarrow \mathbb{R}_+$ where
$$f(p) := \min_{I\in \Gamma} \sum_{T=1}^n p_T\cdot C(I,T)\quad \mbox{and}\quad g(p):=\sum_{T=1}^n p_T\cdot T.$$
Note that when $p$ corresponds to a probability distribution, $f(p)= \min_{I\in \Gamma}\E_{T\gets p}[C(I,T)]$ and $g(p)=\E_{T\gets p}[T]$. So the first expression in the lemma is just $\max_{p\in \cD([n])}\, f(p)/g(p)$. The key observation is that $f$ and $g$ are homogeneous, i.e. $f(a\cdot p)=a\cdot f(p)$ and $g(a\cdot p)=a\cdot g(p)$ for any scalar $a\in \mathbb{R}_+$ and vector $p\in\mathbb{R}_+^n$. This implies that the first expression equals:
$$\max_{p\in \cD([n])}\,\, \frac{f(p)}{g(p)} \quad =\quad \max_{p\in \mathbb{R}_+^n}\,\, \frac{f(p)}{g(p)} \quad =\quad \max_{p\in \mathbb{R}_+^n\,:\,g(p)\le 1}\,\, f(p).$$
It is easy to check that this equals the dual LP above, which proves the lemma.
\end{proof}

Thus the above instance of stochastic knapsack cover has adaptivity gap at least $e-o(1)$.

It can also be shown that every instance of stochastic knapsack cover with a single stochastic item has adaptivity gap  at most $e$: this uses a relation to the incremental $k$-MST problem~\cite{LNRW10}.


\section{Conclusion}
The main open question is to obtain a constant-factor approximation algorithm for stochastic $k$-TSP. Another interesting question is the adaptivity gap, even in the special case of stochastic knapsack cover: the currently known lower bound is $e$ and upper bound is $O(\log^2k)$.

\section*{Acknowledgment} We thank Itai Ashlagi for initial discussions that lead to this problem definition.

\bibliographystyle{plain}
\bibliography{stoch-kc}

\end{document}